\documentclass{article}
\usepackage{iclr2024_conference,times}

\usepackage{amsmath,amsfonts,bm, amsthm}
\usepackage{wrapfig}
\usepackage{graphicx}
\usepackage{float}
\usepackage{algorithm}
\usepackage{bbm}
\usepackage{algorithmicx}
\usepackage[noend]{algpseudocode}
\usepackage{tikz}
\usepackage{tcolorbox}
\usepackage{multirow}
\newcommand{\eat}[1]{}

\usetikzlibrary{positioning, arrows.meta, shapes.geometric, backgrounds, fit, decorations.pathreplacing}

\newtheorem{theorem}{Theorem}

\newtheorem{lemma}[theorem]{Lemma}

\def\eqref#1{equation~\ref{#1}}

\def\1{\bm{1}}

\DeclareMathAlphabet{\mathsfit}{\encodingdefault}{\sfdefault}{m}{sl}
\SetMathAlphabet{\mathsfit}{bold}{\encodingdefault}{\sfdefault}{bx}{n}

\usepackage{hyperref}
\usepackage{url}

\title{Differentially Private Learned Indexes}

\eat{
\author{
    Jianzhang Du\\
    Indiana University\\
    \texttt{du5@iu.edu}
  \and
    {\bf Tilak Mudgal \& Rutvi Rahul Gadre}\\
    UMass Dartmouth\\
    \texttt{\{tmudgal,rgadre\}@umassd.edu}
    \and
    {\bf Yukui Luo}\\
    Department\\
    \texttt{yluo11@binghamton.edu}
    \and
    {\bf Chenghong Wang}\\
    Indiana University\\
    \texttt{cw166@iu.edu}
}
}

\author{Jianzhang Du\thanks{Co-first author, equal contribution.} \\
Indiana University\\
\texttt{du5@iu.edu}\\
\And
{\bf Tilak Mudgal$^{*}$ \& Rutvi Rahul Gadre$^{*}$}\\
UMass Dartmouth\\
\texttt{\{tmudgal,rgadre\}@umassd.edu}\\
\And
{\bf Yukui Luo}\\
Binghamton University\\
\texttt{yluo11@binghamton.edu}\\
\And
{\bf Chenghong Wang}\\
Indiana University\\
    \texttt{cw166@iu.edu}
}

\iclrfinalcopy 
\begin{document}

\maketitle
\begin{abstract}
In this paper, we address the problem of efficiently answering predicate queries on encrypted databases—those secured by Trusted Execution Environments (TEEs), which enable untrusted providers to process encrypted user data without revealing its contents. A common strategy in modern databases to accelerate predicate queries is the use of indexes, which map attribute values (keys) to their corresponding positions in a sorted data array. This allows for fast lookup and retrieval of data subsets that satisfy specific predicates. Unfortunately, indexes cannot be directly applied to encrypted databases due to strong data-dependent leakages. Recent approaches apply differential privacy (DP) to construct noisy indexes that enable faster access to encrypted data while maintaining provable privacy guarantees. However, these methods often suffer from large storage costs, with index sizes typically scaling linearly with the key space. To address this challenge, we propose leveraging learned indexes—a trending technique that repurposes machine learning models as indexing structures—to build more compact DP indexes. Our contributions are threefold: (i) We first propose a strawman method that directly applies DP training techniques to classical learned indexes. We then show that this general paradigm can result in significant utility concerns, with an error no better than existing DP indexes. (ii) We then introduce a completely different paradigm—learning a private index structure on already distorted noisy key-position mappings. Specifically, we propose a range tree-based private mechanism to generate distorted key-position mappings, followed by error-bounded piece-wise linear regression (PLR) models to learn a compact representation of the noisy mapping. (iii) We provide a formal analysis of our DP-PLR based learned indexes. Results show that our DP-PLR can achieve near-lossless processing with bounded overhead. Additionally, DP-PLR significantly reduces index size, from the linear key-size scaling of traditional methods to potentially constant size.\vspace{-2mm}
\end{abstract}
\section{Introduction}\vspace{-2mm}
Over the past decade, there has been a significant increase in the use of cloud computing for data storage and analysis. Its low cost, high availability, scalability, and ease of use make it an appealing option for businesses and scientific research. However, organizations that handle sensitive data, such as hospitals, banks, government agencies, and energy companies, may hesitate to use cloud services due to privacy concerns. The shared nature of cloud resources, coupled with potential vulnerabilities in the privileged software stack, has already led to various privacy breaches~\citep{ibm2024data}. As a result, there is a critical need for robust measures to safeguard data-in-use privacy in the cloud environment. This is essential not only for policy compliance~\citep{assistance2003summary, voigt2017eu}, but also for maintaining public trust and advancing national priorities~\citep{NationalStrategy2022}

This need has given rise to a long line of research in an area known as {\em Encrypted Databases} (EDBs)\citep{eskandarian2017oblidb, wang2021dp, qiu2023doquet}, which enable untrusted cloud providers to manage and process encrypted user data. To achieve this, EDBs leverage Trusted Execution Environments (TEEs)\citep{costan2016intel} to establish secure hardware enclaves on cloud machines, ensuring that any execution within these enclaves remains strongly isolated from the rest of the software stack, including the privileged OS and hypervisors. Users' data is only decrypted and processed inside these enclaves, and remains encrypted and integrity-protected whenever it leaves the enclave. Despite the strong encryption and isolation provided by TEEs, researchers have identified various side-channel threats associated with TEEs, leading to significant real-world data breaches~\citep{kocher2020spectre}. For example, different query processing can result in distinguishable memory access patterns and read/write volumes, which attackers can exploit to reconstruct substantial portions of the data~\citep{kellaris2016generic}, even if it is placed inside a TEE or encrypted elsewhere. As a result, modern EDBs combine TEEs with oblivious algorithms, which implement branchless processing methods and pad the complexity to the worst-case maximum to ensure complete data independence.

\eat{
\begin{figure}[H]
    \centering
    \begin{minipage}{0.3\textwidth}
        \centering
        \includegraphics[width=\linewidth]{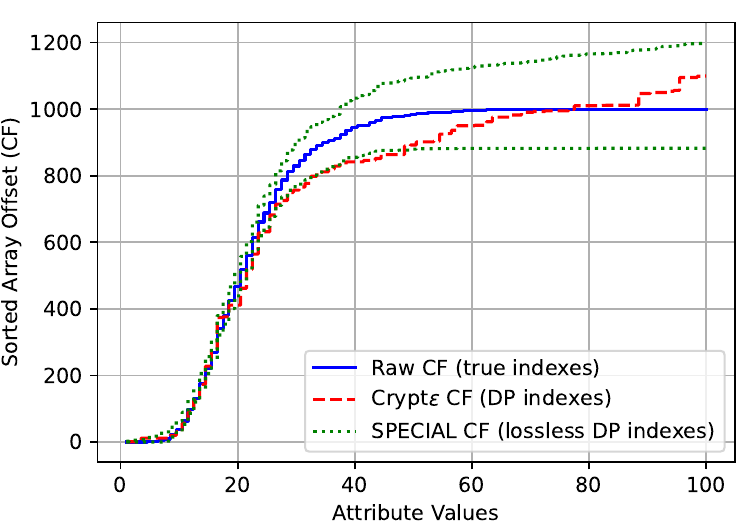} 
        \caption{EDB workflow}
        \label{fig:fig1}
    \end{minipage}
    \begin{minipage}{0.3\textwidth}
        \centering
        \includegraphics[width=\linewidth]{figure/index.pdf} 
        \caption{Indexes}
        \label{fig:fig2}
    \end{minipage}
    \begin{minipage}{0.3\textwidth}
        \centering
        \includegraphics[width=\linewidth]{figure/index.pdf} 
        \caption{DP indexes}
        \label{fig:fig3}
    \end{minipage}
\end{figure}
\vspace{-1em}
}
While oblivious algorithms provide strong and provable privacy guarantees for today's EDBs, they clash with modern database optimization techniques, which often rely on leveraging data-dependent patterns for fine-grained performance improvements. A prime example is the use of indexes, which map attribute values (or keys) to their positions in a sorted array. Indexes enable rapid access to specific data subsets, reducing the need for frequent full table scans and minimizing excessive I/Os. Unfortunately, this promising technique is not directly compatible with EDB's privacy guarantees as they can leak exact information about the data distribution. As a result, EDBs must sequentially load all encrypted data into the TEE for every query processing, even when only a small portion is needed. To bridge this gap,~\citet{sahin2018differentially} and~\citet{chowdhury2019crypt} independently introduced DP indexes, which distort key-position mappings with DP noise to privately index encrypted data. For instance,~\citet{sahin2018differentially} use a $B^{+}$ tree-based index, where (pointers to) private data are stored at leaf nodes, each managing data with the same attribute value. They then apply DP by randomly adding or removing tuples at leaf nodes to create plausible deniability. \citet{chowdhury2019crypt} propose a more storage-efficient design without dummy data, which directly distorts index endpoints. For example, if the true index range for an attribute value is \( D[v_0, v_1) \), the DP index generates \( D[\tilde{v}_0, \tilde{v}_1) \), where each endpoint $\tilde{v}_0$ and $\tilde{v}_1$ is distorted by independent DP noise. Although \( D[\tilde{v}_0, \tilde{v}_1) \) may include irrelevant data, a TEE can filter it inside the enclave. However, the symmetric DP noise can still cause data loss when \( D[v_0, v_1) \setminus D[\tilde{v}_0, \tilde{v}_1) \neq \emptyset \). To address this,~\citet{wang2024special} use one-sided DP noise to ensure lossless indexing by enforcing \( \tilde{v}_1 > v_1 \) and \( \tilde{v}_0 < v_0 \). Nevertheless, all these DP indexes share a common limitation: storage costs typically \( O(N) \), where \( N \) is the number of keys. For instance, a $B^{+}$ tree requires at least $2N-1$ nodes, and methods like~\citep{roy2020crypt, wang2024special} will store at least one noisy endpoint per attribute value.

\citet{kraska2018case} argue that indexes are inherently models, sparking a growing research area known as learned indexes~\citep{wu2024automatic}, which repurpose machine learning (ML) models as database indexes. Recent advancements in learned indexes have demonstrated their powerful storage efficiency, showing that even with a limited number of models, accurate key-position mappings can be approximated. When models are simple, such as linear models, the total storage for learned indexes can be reduced to constant sizes. This leads to the fundamental question of this work:
\begin{center}
    {\em Can we leverage learned index techniques to build new DP indexes for EDBs that are both provably private and compact in storage size?}
\end{center}

To address this question, we initiate the first study on designing DP learned indexes. Our major contributions are as follows: (i) We begin by proposing a strawman method that directly applies existing DP training techniques to the training of classical learned indexes. Nevertheless, we show that this general paradigm can lead to significant utility concerns due to the large sensitivity in gradients. We provide an empirical risk lower bound for this method, demonstrating that it performs no better than existing DP indexes. (ii) To achieve practical DP learned indexes with strong utility, we advocate for a completely different paradigm—applying standard training on distorted key-position mappings. To accomplish this, we propose a range tree-based mechanism that generates noisy key-position mappings, with privacy-induced error bounded by \smash{$O(\log^{3/2}{N})$}. We then apply error-bounded piece-wise linear regression (PLR) to learn a compact representation of these noisy mappings. (iii) We provide a formal analysis of our proposed DP-PLR method and compare its guarantees with existing DP indexes. Results show that the PLR yields index sizes in $O(1)$ and can, with high probability, achieve lossless indexing, while keeping the index lookup overhead (the number of irrelevant tuples indexed) bounded by \smash{$O(\log^{3/2}{N})$}.

\section{Background and Related Work}
\vspace{-1mm}\noindent{\bf General notations for relational databases.} We define {\em database instances} \( D \) as relational tables with attributes \( \texttt{attr}(D) \), where each attribute \( A \in \texttt{attr}(D) \) has a discrete domain of \( \texttt{dom}(A)=\{x_i\}_{i=1}^{N} \). By default, we consider $\texttt{dom}(A)$ is sorted by ascending order. 
We define the {\em frequency (count)} of an attribute value $x$ of $A$ in $D$ as $\texttt{F}(x, D, A) = \sum_{t\in D \wedge t.A = x} \mathbbm{1}$. The frequencies for all attribute values of $A$, is defined as the histogram of $A$ in $D$, denoted as $\texttt{H}(D, A_i) = \{c_k=\texttt{F}(x_k, D, A_i)\}_{\forall x_k \in \texttt{dom}(A_i)}$. For simplicity, we may omit the arguments $D$ and $A$ of $\texttt{F}(\cdot), \texttt{H}(\cdot)$ if they're already defined. In this work, we focus on linear predicate queries, denoted as \( q_{\phi}(D) \), which retrieve tuples from \( D \) satisfying a predicate \(\phi\) then compute aggregated statistics on the fetched data. A predicate \(\phi\) is a logical expression with conditions on attributes, formed using conjunctions ($\land$) or disjunctions ($\lor$). Each condition is a logical comparison such as \( A_i = a \) or \( A_j > b \). 

\noindent{\bf EDB system model.} We consider a standard EDB model~\citep{zheng2017opaque, eskandarian2017oblidb} in a cloud environment with two entities (as shown in Figure~\ref{fig:fig1}): the service provider (SP), managing the cloud infrastructure including the TEE, and the data owner (DO), who securely outsources storage and processing of private data to the SP. In the standard EDB model, the SP is considered honest-but-curious~\citep{paverd2014modelling}, meaning they follow the pre-defined EDB protocol without deviation but may attempt to learn sensitive information about the owner's data by observing execution transcripts.

To initiate the EDB, the TEE first creates an enclave on SP's cloud machine, generates encryption keys $(\mathsf{sk}, \mathsf{pk})$, keeps $\mathsf{sk}$ inside the enclave, and sends $\mathsf{pk}$ to the owner. The DO encrypts their data tuple-wisely using $\mathsf{pk}$ along with a result key $\mathsf{K}{\mathsf{r}}$, then uploads the encrypted data and $\mathsf{K}{\mathsf{r}}$ to the SP. When the DO issues a query, the TEE loads, decrypts, and processes the data, reencrypts the result using decrypted $\mathsf{K}_{\mathsf{r}}$, and sends it back to the DO for decryption. Moreover, we assume the EDB employs oblivious algorithms for data processing to prevent side-channel leakages. For predicate queries, an example of the oblivious processing algorithm is as follows: (i) Upon receiving the query, the TEE sequentially reads the entire dataset into the enclave and performs a linear scan, labeling each tuple as a "match" or "non-match" based on the predicate. To maintain obliviousness, the label is updated for every tuple, regardless of whether it matches the predicate; (ii) The TEE then makes a second pass over the labeled data to compute the desired statistics as specified by the query. Similarly, during the statistics computation, every tuple is accessed, including non-matching ones, for which a dummy read is applied to maintain obliviousness.

\vspace{-1em}
\begin{figure}[H]
    \centering
    \begin{minipage}{0.35\textwidth}
        \centering
        \includegraphics[width=\linewidth]{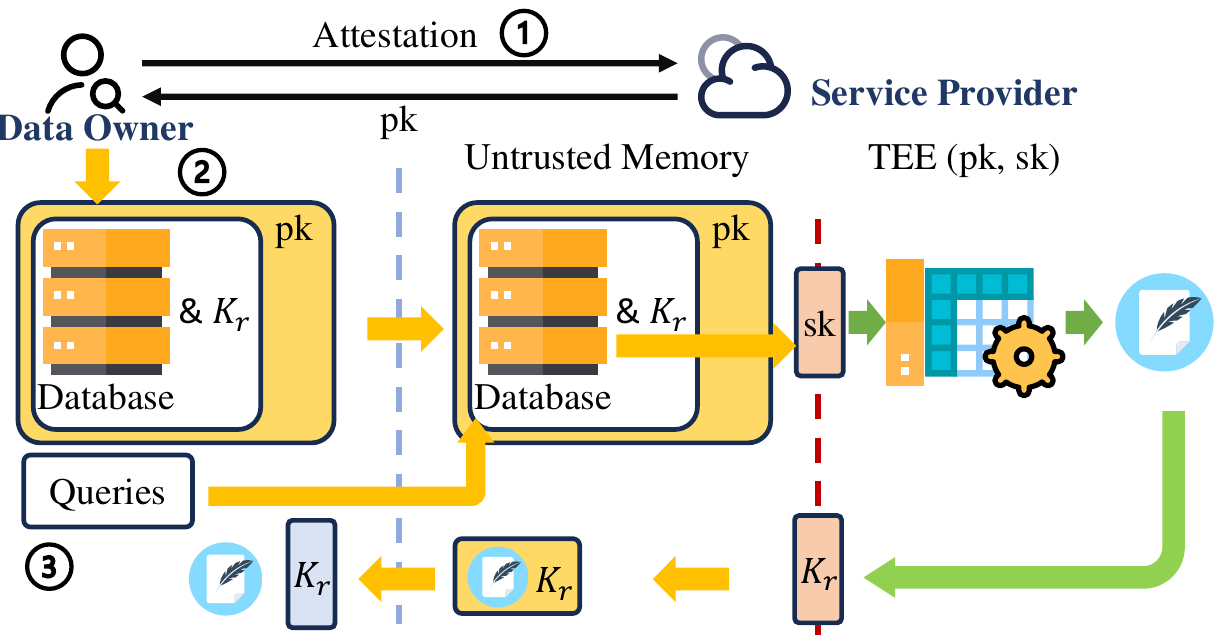} 
        \caption{EDB workflow}
        \label{fig:fig1}
    \end{minipage}%
    \begin{minipage}{0.3\textwidth}
        \centering
        \includegraphics[width=\linewidth]{figure/index.pdf} 
        \caption{CF model indexes}
        \label{fig:fig3}
    \end{minipage}
    \begin{minipage}{0.34\textwidth}
        \centering
        \includegraphics[width=\linewidth]{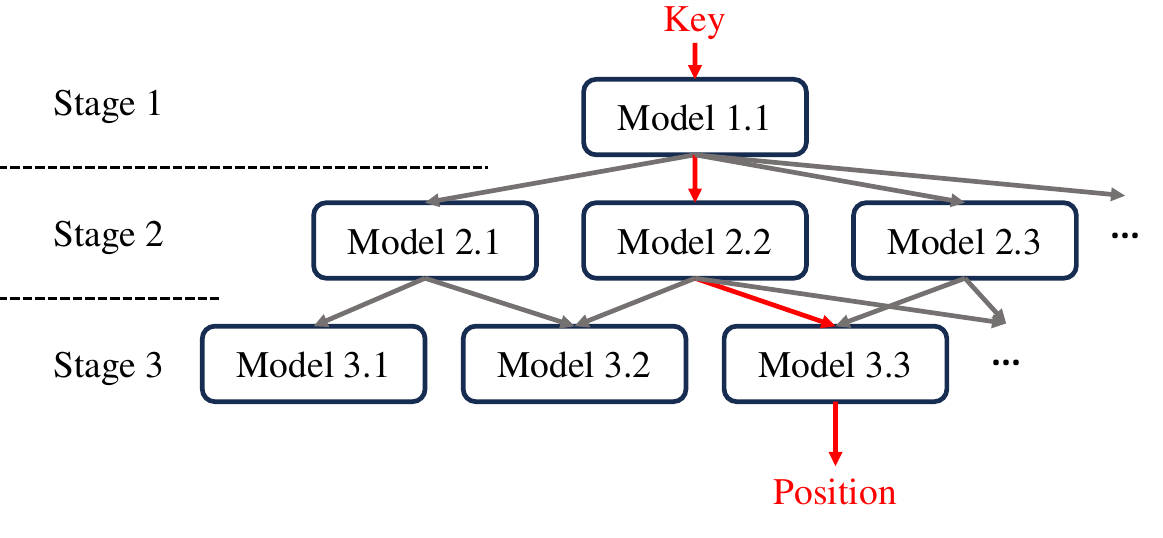} 
        \caption{RMI example}
        \label{fig:fig3}
    \end{minipage}
\end{figure}
\vspace{-1em}
\noindent{\bf Index model and learned indexes.} Given $D$ sorted by $A$, and an attributed value $x_i \in \texttt{dom}(A) = \{x_1, x_2, \dots, x_n\}$, a index structure maps $x_i$ to an interval $[v_0, v_1)$ such that $D[v_0, v_1)$ contains all tuples $t\in D$ that $t.A = x_i$.~\citet{kraska2018case} suggest that indexes can essentially be abstracted as a cumulative frequency (CF) model. For instance, the lower and upper bound of $[v_0, v_1)$ mentioned above can actually computed by the CFs of $v_0=\texttt{CF}(x_{i-1})=\sum^{i-1}_{j=0} \texttt{F}(x_j)$, and $v_1=\texttt{CF}(x_i)$\footnote{With out loss of generality, we implicitly assume $\texttt{F}(x_0)=0$.}. To better show this, we provide a visual example in Figure~\ref{fig:fig3}. \citet{kraska2018case} argue that traditional indexes (e.g., B\(^+\) trees) are inherently data structures that compute or approximate a CF curve (CFC), which is an ordered collection \(\{x_i, y_i\}_{i=1}^{N}\), where \(y_i = \texttt{CF}(x_i)\). They propose that these classical data structures can be replaced by ML models, namely learned indexes, which learn to approximate the CFC. A well-known example of a learned index is the Recursive Model Index (RMI)~\citet{kraska2018case}. The RMI mimics the structure of a $B^{+}$ tree but replaces its internal nodes with models. For a lookup operation, the root model predicts which child model should handle the request, and this process continues through successive stages until a leaf model is reached. The final (leaf) model then estimates the key’s position by fitting it to the CFC (Figure~\ref{fig:fig3}).



\noindent{\bf Differential privacy.} 
DP is a well-established privacy framework that is rooted from the property of algorithm stability. Specifically, a randomnized mechanism $\mathcal{M}$ is said to satisfies $(\epsilon, \delta)$-DP if for any pair of neighboring databases $D$ and $D'$, differing by at most one tuple, and for all $\forall O \subset \mathcal{O}$, where $\mathcal{O}$ denotes all possible outputs, the following holds
$$\textup{Pr}[\mathcal{M}(D)\in O] \leq e^{\epsilon}\textup{Pr}[\mathcal{M}(D')\in O] + \delta$$
 This definition ensures that the probability of producing a specific output does not change significantly (up to a multiplicative factor of $e^{\epsilon}$) when any single tuple in the dataset is modified. The slack $\delta$ introduces a practical relaxation, allowing the privacy guarantee to fail with probability at most $\delta$. Note that when $\delta=0$, then the above definition is commonly referred as $(\epsilon)$-DP or pure DP.
 

\noindent{\bf DP indexes.}  Two independent works~\citet{sahin2018differentially, chowdhury2019crypt} introduced the concept of DP indexes to enable EDBs to use indexing techniques. Specifically, \citet{sahin2018differentially} use a variant $B^{+}$ to implement DP indexes, where a binary tree routes keys to a leaf that stores (pointers to) all records with a specific attribute value. To ensure DP, they adjust the leaf data by adding (positive noise) or removing (negative noise) tuples based on randomly sampled DP noise. To prevent important records from being removed due to negative noise, a pessimistic amount of dummy tuples is added to each leaf data before the DP adjustment. To answer predicate queries, the entire tree is searched, and all leaf data, including dummy ones, within the predicate range is returned. While this method ensures basic DP guarantees, it results in significant storage costs due to the large amount of dummy data involved. 

\citet{chowdhury2019crypt} proposed the Crypt$\epsilon$ that do not require dummy data. Specifically, they consider to directly release a noisy CFC $\{x_i, \tilde{y}_i\}_{i=1}^{N}$ such that each $\tilde{y}_i\gets \texttt{CF}(x_i)+\texttt{Lap}(\frac{N}{\epsilon})$ is distorted by independent Laplace noise. For processing queries, the EDB first obliviously sorts the entire data by the attribute, then uses the noisy CF to index tuples. We show an example of Crypt$\epsilon$ in Figure~\ref{fig:fig3}. While Crypt$\epsilon$ indexes avoid adding dummy data, they suffer from significant data loss. The noisy index range, distorted by DP, may be smaller than the true range, potentially excluding a large volume of matching tuples.

\cite{wang2024special} improves upon Crypt$\epsilon$ by introducing SPECIAL indexes. Specifically, they create two noisy CFCs: one that overestimates the true CFC with one-sided positive Laplace noise noise and another that underestimates it with negative noise. To answer index lookup, the lower end point is taken from the underestimated curve, and the upper end point from the overestimated curve, ensuring all matching data is included. However, this approach can result in significant overhead, as the noisy index range can be much larger than the true range, leading to substantial I/O costs. We shown an example of this in Figure~\ref{fig:fig3}. To our knowledge, SPECIAL index is the only project that guarantees deterministic accuracy (no missing data).


\section{DP Learned Indexes}
In this section, we present the technical details of our proposed DP learned indexes. Before delving into the specifics, we first formulate the concrete problem to be addressed.

\noindent{\bf Problem formulation.} We consider the problem of private training and inference of learned indexes on static data. Formally, given $D$ sorted by $A \in \texttt{attr}(D)$, and $\{(x_i, y_i)\}_{i=1}^N$ to be the CFC of $A$ over $D$. We consider an idealized index to be $\textsc{Idx}(x_i) = [y_{i-1}, y_{i})$ for all $x_i$, and our goal is to build a private learned index model, $\textsc{PIdx}$, such that $\textsc{PIdx}(x_i) = [\hat{y}(x_{i-1}), \hat{y}(x_{i}))$, where $\hat{y}(x_{i})$ denotes the predicted position of $x_{i}$. In addition, we consider $\textsc{PIdx}$ should satisfies the following properties:
\begin{itemize}
    \item {\bf P-1. $(\epsilon, \delta)$-DP at the tuple level.} Given $\epsilon > 0$, and $0\leq \delta<1$. For any neighboring databases \(D\) and \(D'\), differing by a single tuple and both sorted by the same attribute \(A_i\), where \(\texttt{dom}(A_i)\in D = \texttt{dom}(A_i)\in D'\), it holds for all outputs \(O \subset \mathcal{O}\) that 
\[
\textup{Pr}[\textsc{PIdx}^{D} \in O] \leq e^{\epsilon} \cdot \textup{Pr}[\textsc{PIdx}^{D'} \in O] + \delta.
\]
This notion describes privacy at the tuple level, for instance the information related to the presence of each individual tuple by observing the index outcome, is bounded by DP. We do not consider key privacy, for example, an attacker might know whether a specific key \(x_j\) is included in \(\texttt{dom}(A_i)\in D\). This aligns with the privacy guarantees of existing DP index approaches~\citep{roy2020crypt,wang2024special} and the logic of traditional indexes, where a lookup should abort if a key is not present. Note that tuple-level privacy can also be extended to user-level privacy. If a user owns multiple tuples in the dataset and we wish to preserve privacy at the user level, we can achieve this by applying the group privacy~\citep{dwork2016concentrated} mechanism with appropriately adjusted privacy parameters. 
\item {\bf P-2. Practical utility.} We require $\textsc{PIdx}$ should offer practical utility guarantees, including high accuracy (e.g., for any $x_i$, $|\textsc{Idx}(x_i) \setminus \textsc{PIdx}(x_i)|$ is small), and low query overhead ($\forall x_i, |\textsc{PIdx}(x_i) \setminus \textsc{Idx}(x_i)|$ is small).

\item {\bf P-3. Storage efficient.} One key motivation for exploring learned indexes is the potential to use lightweight ML models instead of large data structures. Therefore, we believe DP indexes should retain this advantage, offering more storage efficiency than traditional methods, such as noisy CF-based classical DP indexes.
\end{itemize}

\subsection{Na\"ive designs}\label{sec:challenges}
 A straightforward approach is to directly apply private training methods, such as DP-SGD, to RMIs. However, we notice that the hierarchical training process considered by RMIs can cause significant utility loss when paired with DP-SGD. Specifically, consider an RMI with \( \ell \) stages, comprising one root model and \( M_{\ell} \) sub-models at each stage \( \ell \). We use \smash{\( f_0(x; \theta_0) \)} to denote the root model, where \( \theta_0 \) are its parameters. A sub-model at stage \( \ell \), indexed by \( k \), is denoted by \smash{\( f_{\ell}^{(k)}(x; \theta_{\ell}^{(k)}) \)}, with corresponding parameters \smash{\( \theta_{\ell}^{(k)} \)}. The training of the RMI proceeds sequentially, stage by stage. A model in stage \( \ell \) cannot begin training until all models in stage \( \ell-1 \) have been trained. Each model at stage \( \ell \) is trained on a subset of data that is routed by the models from the previous stage \( \ell-1 \). For example, the sub-model \smash{\( f_{\ell}^{(1)} \)} at stage \( \ell \) will only receive the training data where the model in stage \( \ell-1 \) predicts a value that routes the data to sub-model \( k=1 \). Formally, the training data for sub-model \smash{\( f_{\ell}^{(k)} \)} is given by:
\[
CF_{\ell}^{(k)} = \left\{ (x_i, y_i) \in CF \,\middle|\, \left\lfloor \frac{M_{\ell} f_{\ell-1}(x_i; \theta_{\ell-1})}{N} \right\rfloor = k \right\}.
\]
In this context, \smash{\( f_{\ell-1}(x_i; \theta_{\ell-1}) \)} refers to the prediction from the previous stage and determines which sub-model \smash{\( f_{\ell}^{(k)} \)} will be responsible for training on the given data point \( (x_i, y_i) \). To incorporate DP-SGD into the aforementioned training process, independent DP noise must be added at each stage, however, this can potentially lead to significant error accumulation in the final-stage models. Moreover, tracking and managing the privacy budget across different stages becomes challenging. 

To address the challenge that traditional RMI structures can be hard to compatible with DP-SGD, we draw inspiration from the Mixture of Experts framework to introduce FMIs. Unlike the hierarchical structure of classical RMI, which necessitates independent training of each sub-model, the FMI consists of multiple models that can be trained at the same time, which significantly simplifies integration with DP-SGD---as one can derive a unified gradient during training to which DP noise can be applied. Specifically, we consider a FMI consists of a root model \( f_0(x; \theta_0) \) and \( M \) expert-models \( \{f_k(x; \theta_k)\}_{k=1}^M \) operating at the same level. Given a key $x_i \in  CF = \{(x_i, y_i)\}_{i=1}^N $, the index prediction of FMI is then given by
\vspace{-1em}
\[
\textsc{PIdx}(x_i) = \left(\hat{y}(x_{i-1}), \hat{y}(x_{i})\right)~~s.t.~~\hat{y}(x) = \sum_{k=1}^{M} w_k(x; \theta_0) \cdot f_k(x; \theta_k),
\]
where \( \mathbf{w}(x; \theta_0) = [w_1(x; \theta_0), \dots, w_M(x; \theta_0)] \) denotes the weight vector generated by the root model. We say that FMI can be logically similar to RMI: FMI can be logically formed as a two staged RMI, but instead of picking one model for prediction, multiple last stage models collaboratively predict the output for every key, with a root model $f_0(x;\theta_0)$ adjusting the contributions of these predictions through a weight vector $\mathbf{w}(x_i; \theta_0)$. Notably, if we generalize the root model to output a one-hot-encoded weight vector, the FMI effectively reduces to a two-stage RMI model.

To enable simultaneous training of all components, the model parameters are concatenated into a unified parameter vector $\Theta = [\theta_0; \theta_1; \dots; \theta_M]$. 
During each training iteration, we sample a mini-batch \( B \subset D \), where the per-example unified gradient is computed as
\[
\begin{aligned}
    g_i &= \nabla_{\Theta} \ell_i(\Theta) = \begin{bmatrix}
        \nabla_{\theta_0} \ell_i(\Theta), \nabla_{\theta_1} \ell_i(\Theta), \dots
        \nabla_{\theta_M} \ell_i(\Theta)
    \end{bmatrix}^{\top},
\end{aligned}
\]

In this way, a unified per-example gradient vector \( g_i \) can be obtained. Next, we apply DP-SGD to update the overall model parameters, and using the obtained unified gradients, shown as follows:
\[
\begin{aligned}
    \Theta &\leftarrow  \Theta - \eta \left( \frac{1}{|B|} \left( \sum_{i \in B} \frac{g_i}{\max\left(1, \frac{\| g_i \|_2}{C} \right)} + \mathcal{N}\left(0, \sigma^2 {|B|^2}C^2 I \right) \right) \right),
\end{aligned}
\]
The above process can be summarized as follows: At each iteration, we sample a mini-batch $B\subset\{(x_i, y_i)\}_{i=1}^N$ for training. For each example in the mini-batch, the per-example gradient \( g_i\) is first clipped to a maximum \( \ell_2 \)-norm of \( C \) to bound the sensitivity. After that, the clipped gradients are averaged and perturbed with Gaussian noise \( \mathcal{N}(0, \sigma^2|B|^2 C^2 I) \), where \( I \) is the identity matrix corresponding to the dimensionality of \( \Theta \), and \( \sigma >0 \) is the noise multiplier. Finally, the DP model update is applied to the overall model parameters with a learning rate \( \eta \). 

Although FMI enables seamless integration with DP-SGD, a key limitation lies in the high sensitivity of gradients. Specifically, altering a single tuple in the dataset \(D\) can affect all entries in \(\{x_i, y_i\}_{i=1}^{N}\). This, in turn, may cause every per-example gradient in a mini-batch \(B\) to change, resulting in a sensitivity of \(|B|\). Consequently, Gaussian noise must be scaled by the batch size, which significantly reduces the utility of DP-SGD and can lead to poorly trained models. For example, applying the excess risk lower bounds of DP-SGD~\citep{bassily2014private}, and accounting for the Gaussian noise multiplier \(\sigma\) being scaled by \(|B|\), the excess risk is lower bounded by \smash{$\Omega\left(\frac{p}{N (\epsilon / |B|)^2}\right)$} even for strongly convex functions, where \(p\) is the dimensionality of the model. Furthermore, since the batch size \(|B|\) is typically proportional to the sample size, i.e., \(|B| = O(N)\), the error can be lower bounded by \smash{$\Omega\left(\frac{Np}{\epsilon^2}\right)$}. In contrast, classical DP indexes, such as Crypt\(\epsilon\), bound the privacy-induced error with high probability at \smash{$O\left(\frac{\sqrt{N}}{\epsilon}\right)$} 
(see~\citet{roy2020crypt} or our formal analysis in Table~\ref{tab:formal}). As such, the utility of DP-SGD learned indexes would have a lower bound that is asymptotically worse than that of traditional DP indexes.

\subsection{Proposed method}\label{sec:proposed}
To prevent large utility loss, we advocate a completely different paradigm for building DP learned indexes: instead of training on the original CFCs, we train on CFCs that have already been distorted by DP mechanisms. Since the training is performed as post-processing on the DP-distorted CFCs, no additional noise is required. However, two key challenges remain. First, the high sensitivity of CFCs can result in significant noise being introduced when releasing the noisy CFCs, leading to substantial utility loss in the final learned indexes, even if the training is noise-free. Second, model selection can be crucial, especially when one of our primary goals is achieving storage efficiency (P-3). Hence, we must carefully choose models that can learn a compact representation of the noisy CFC while maintain high utility. Below, we discuss our approaches to address these challenges.\vspace{-.5em}
\begin{figure}[!h]
    \centering    \includegraphics[width=0.85\linewidth]{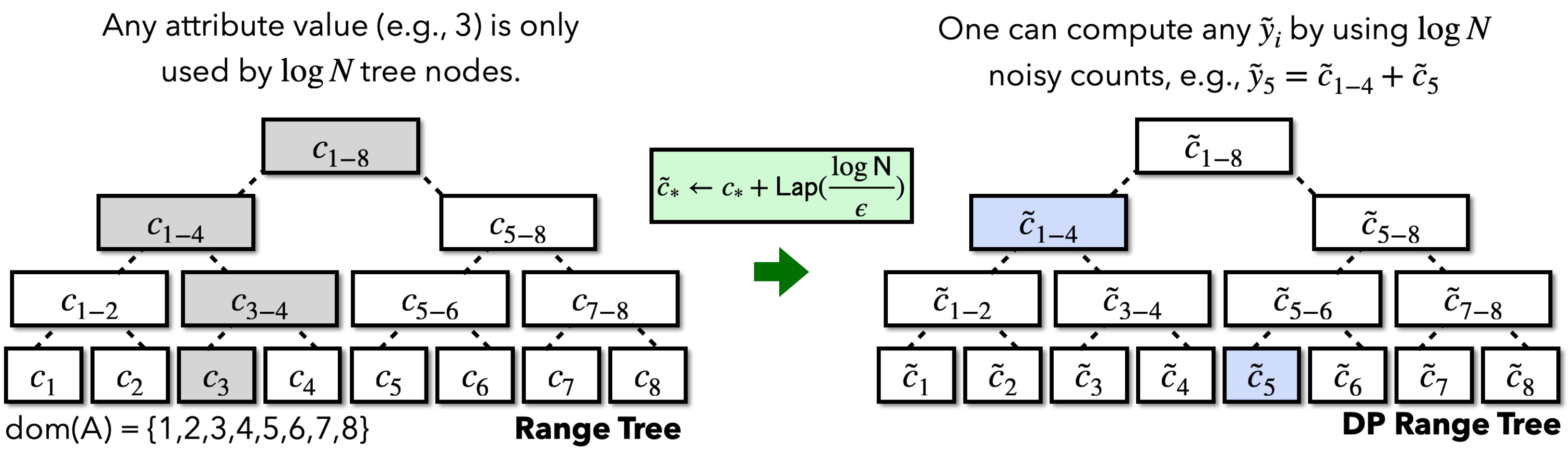}
    \caption{Range tree and DP range tree.}
    \label{fig:range}
    \vspace{-1em}
\end{figure}

{\bf Range tree based noisy CFC generation.} To address the first challenge, we drew inspiration from prior work on the continual release of time-series data~\citep{chan2011private,wang2023private} and propose a similar range tree-based approach for deriving noisy CFCs. This method allow us to bound the DP induced errors within $O((\log{N})^{3/2})$. 
At a high level, we construct a binary range tree over the sorted attribute domain $\texttt{dom}(A) = \{x_i\}^{N}_{i=1}$. Each leaf node represents a single attribute value $x_i \in \texttt{dom}(A)$ and stores the frequency of $x_i$. Each internal node represents the combined attribute range of its two child nodes and stores the frequency of records whose attribute
$A$ falls within that range. The tree construction proceeds upward recursively until we reach the root node, which represents the entire attribute range $[x_1, x_N]$. We show an example in Figure~\ref{fig:range}. Next, we convert this tree into a DP range tree by adding Laplace noise to each frequency count stored in every node. Since tuples with a certain attribute value contribute to at most \(\log{N}\) frequency counts in the tree, thus by the parallel composition theorem of DP~\citep{dwork2014algorithmic}, the per-instance DP noise can be set as \(\mathsf{Lap}(\frac{\log{N}}{\epsilon})\). This ensures that after adding noise to all frequency counts, the entire mechanism still satisfies \(\epsilon\)-DP. Note that each $\tilde{y}_i$ can be derived using at most $\log{N}$ noisy frequencies from the DP range tree, with each acting as a partial sum (p-sum) contributing to the final value of \(\tilde{y}_i\). We show an example of this approach in Figure~\ref{fig:range} (right).

\begin{theorem}\label{tm:bound} For any CFC $\{x_i, y_i\}_{i=1}^{N}$, and let $\{x_i, \tilde{y}_i\}_{i=1}^N$ to be the corresponding noisy CFC released by the range tree mechanism. Then, $\forall y_i$, the error $|\tilde{y}_i - y_i|$ is bounded by {$O(\epsilon^{-1}{(\log{N})^{\frac{3}{2}}})$.}
\end{theorem}
The privacy-induced error in the range-tree mechanism is determined by the combined effect of Laplace noise. Since each \(\tilde{y}_i\) is computed using at most \(\log{N}\) noisy counts, each perturbed by \(\mathsf{Lap}(\frac{\log{N}}{\epsilon})\) noise, hence, Theorem~\ref{tm:bound} can be proven by applying the bound on the sum of Laplace random variables~\citep{dwork2014algorithmic, wang2021dp}. In algorithm~\ref{alg:binary_mechanism}, we outline the formal procedure for implementing the range-tree based noisy CFC releasing, including additional optimizations for dynamically recycling unused noisy frequencies.
\begin{algorithm}
\begin{algorithmic}[1]
\Statex
\textbf{Input}: $D$ sorted by $A$; Privacy parameter \( \epsilon > 0 \).
\smallskip
\State Compute histogram $H(D, A) = \{c_i=F(x_i, D, A)\}^{N}_{i=1}$, and $\{x_i\}_{i=1}^{N} = \texttt{dom}(A)$
\State Set p-sum vectors $\{\alpha_i=0\}^{N}_{i=1}$, $\{\tilde{\alpha}_i=0\}^{N}_{i=1}$, and per instance privacy budget \( \epsilon' = \epsilon / \log_2 N \).
\For{each \( t = 1, 2, \ldots, N \)}
    \State Write $t$ in binary form \( t = \sum_j 2^j\mathsf{bin}_j(t) \), where \( \mathsf{bin}_j(t) \in \{0, 1\} \).
    \State Let \( i = \min\{ j : \text{Bin}_j(t) \neq 0 \} \), update the p-sum $\alpha_i = \sum_{j < i} \alpha_j + c_t$
    \For{\( j < i \)} reset \( \alpha_j = 0 \) and \( \tilde{\alpha}_j = 0 \) \Comment{Dynamic recycle unused p-sums.}
    \EndFor
    \State Generate noisy p-sum: \( \tilde{\alpha}_i = \alpha_i + \text{Lap}\left( \frac{1}{\epsilon'} \right) \).
    \State $\tilde{y}_t = \sum_{j: \text{Bin}_j(t) = 1} \tilde{\alpha}_j$ \Comment{Find noisy p-sums for covering $\tilde{y}_t$.}
\EndFor
\State 	{\bf return} \(\{(x_i,\tilde{y}_i)\}_{i=1}^N \).
\end{algorithmic}
\caption{Optimized noisy range tree mechanism for generating noisy CFC}
\label{alg:binary_mechanism}
\end{algorithm}

\noindent{\bf Piece-wise linear regression (PLR) on noisy CFCs.} Once we have the noisy CFC, the next step is to find proper model to learn a compact representation of this CFC. Specifically, we select PLR models for this task, for two primary reasons. First, as the CFC is monotonic and low-dimensional (2d integer data), a PLR model can straightforward fit the trend without overcomplicating the representation. This allows us to represent the CFC with a minimal number of parameters, ensuring model simplicity and storage efficiency (P-3). Second, recent advancements in error-bounded PLR~\citep{xie2014maximum} allow one to approximate each data point in a given dataset within a fixed error bound \( \tau \). This enables us to enforce a utility goal by controlling the approximation accuracy. Furthermore, one can also manage the model complexity by adjusting the error bound \( \tau \); smaller values of \( \tau \) typically result in fewer segments and, consequently, a smaller model size. This flexibility allows us to balance between utility (P-2) and storage efficiency (P-3) goals, depending on the specific application requirements. We summarize key steps for applying PLR on noisy CFC in Algorithm~\ref{alg:piecewise_regression}.

\begin{algorithm}
\begin{algorithmic}[1]
\Statex
\textbf{Input}: A sequence of noisy CFC \(\{(x_i, \tilde{y}_i)\}_{i=1}^N \).
\smallskip
\State  \( \{x_i, \tilde{y}_i\}^{N}_{i=1} \gets \texttt{Isotonic\_Regression} (\{x_i, \tilde{y}_i\}^{N}_{i=1}) \) \Comment{monotonic increasing.}
\For{\( i = 1, 2, \ldots, N \)} \( \tilde{y}_i = \min(\max(0, \tilde{y}_i), |D|) \)\Comment{clipping.}
\EndFor
\State $(\texttt{seg} = (v_1, v_2,..., v_k=x_N), f=\{f_i\}^{k}_{i=1}) \gets \texttt{Err\_Bounded\_PLR} (\{x_i, \tilde{y}_i\}^{N}_{i=1}, \tau)$
\State {\bf return} $(\texttt{seg}, f)$
\end{algorithmic}
\caption{Linear model learned indexes from noisy CFCs}
\label{alg:piecewise_regression}
\end{algorithm}

In general, we first apply constrained post-processing to the released noisy CFC to ensure that the noisy CFC adheres to its meaningful context. Specifically, this includes applying isotonic regression to enforce the monotonic increasing property of the CFC, followed by a clipping to ensure that all CF values fall within the valid range $(0, |D|]$. Next, we adopt PLR on the processed CFC, which segments the domain \( \texttt{dom}(A_i) \) into a limited number of disjoint intervals \( [v_{j}, v_{j+1}) \), where \( j = 1, 2, \ldots, M \) and \( M \) is the total number of segments. For each segment \( j \), an individual linear model \( f_j(x) = a_j x + b_j \) is learned to approximate the subset of the CFC within the interval \( [v_{j}, v_{j+1}) \) by minimizing the squared error between the model's predictions and the noisy CFC values. Finally, we output the overall partition, \( \texttt{seg} = (v_1, v_2, \ldots, v_M) \), as well as the set of linear models for each segment, \( f = \{f_i\}_{i=1}^{M} \). Given any input \( x_i \in \texttt{dom}(A) \), to make predictions, the PLR model outputs
\[
f(x_i) = f_j(x_i) \quad \text{for } x_i \in [v_j, v_{j+1}), \quad \text{where } j < M.
\]

\noindent{\bf Nearly lossless index lookup.}
We now show how to make index lookup using the learned PLR. A straightforward approach could be directly use the PLR's predicted position for constructing indexing ranges, for instance, for any $x_i \in \texttt{dom}(A)$, we output $\textsc{PIdx}(x_i)=[f(x_{i-1}), f(x_i))$. Note that due to the potential noises from both DP CFC generation and the PLR learning process, it is possible that $[y_{i-1}, y_i) \setminus [f(x_{i-1}), f(x_i))\neq \emptyset$ and thus leading to a set of matching data to be lost if we follow  $\textsc{PIdx}(x_i)$ to fetch data. To address this issue, we propose a pessimistic indexing method such that with high probability, it holds that $[y_{i-1}, y_i) \subseteq \textsc{PIdx}(x_i)$. Specifically, we first compute the inference errors, $e_i = |f(x_i) - \tilde{y}_i|$ of PLR on noisy CFC, for all $x_i \in \texttt{dom}(A)$, and subsequently determine the maximum error, $e_{\max} = \arg\max_{\{1,...,N\}}(e_i)$. When using error-bounded PLR, one can approximate \( e_{\max} \) directly by the error bound \( \tau \) without additional computation. However, in practice, the empirical \( e_{\max} \) can be much smaller than \( \tau \)~\citep{xie2014maximum}, so tha it is important to compute the exact \( e_{\max} \). Then, for any index lookup, say with key $x_i$, we compute a pessimistic overestimated indexing range as 
$$\textsc{PIdx}(x_i) = \left[\min\left(0, f(x_{i-1}) - e_{\max} - Z\right), \max(|D|, f(x_i) + e_{\max} + Z)\right)$$
where $Z = \alpha_{s}\epsilon^{-1}(\log{N})^{3/2}$ where $\alpha_{s}\geq1$. Since $Z$ is proportional to the noisy CFC's error bound, and the inference error of PLR (on noisy CFC) is bounded by $e_{\max}$. Hence, when $\alpha_{s}$ is set to properly large, then it holds $\textsc{PIdx}$ covers $[y_{i-1}, y_i)$ with high probability, which implies strong lossless indexing guarantees. Moreover, as $e_{\max}$ is computed entirely as a post-processing step on the noisy CFC, and $Z$ is determined by public parameters. As a result, this process does not introduce any privacy loss.

\section{Formal analysis}
In this section, we provide a formal analysis of our proposed DP PLR indexes and compare them against the SOTA DP indexes. Our focus will be on their critical guarantees including privacy (P-1), utility (P-2) and storage efficiency (P-3).

\textbf{Comparison baselines.} We select three existing DP indexes for comparison: the DP $B^{+}$ tree~\citep{sahin2018differentially}, Crypt$\epsilon$~\citep{chowdhury2019crypt}, and SPECIAL indexes~\citep{wang2024special}. Since our focus is on static data indexes, for fair comparison reasons, we exclude DP index variants that are optimized primarily for dynamic data~\citep{zhang2023longshot}. 

{\bf General Setting.} We consider a sorted database $D$ by attribute $A$, with the true CFC represented as $\{x_i, y_i\}_{i=1}^{N}$. Our analysis evaluates four key metrics related to index lookups for answering predicate queries. Specifically, we focus on a predicate query that retrieves data matching a specific attribute value (e.g., point query). We note that the techniques presented here can be extended to other types of predicate queries, such as range queries or arbitrary conjunctive queries, suggesting that our analysis of point queries is sufficient. Specifically, we consider the following metrics: (i) {\em Query error}, which represents the total number of missing data tuples for the index lookup; (ii) {\em Query overhead}, indicating the total number of irrelevant tuples indexed; (iii) {\em Index storage}, which measures the storage required for the index structures; and (iv) {\em Data overhead}, reflecting the number of dummy tuples that must be inserted to the underlying data to support the proposed index. Our analysis will focus on the probabilistic upper bounds, say $\alpha$, for each of the aforementioned measures, such that with probability at least $1 - \beta$ (for $\beta>0$), these measures do not exceed the corresponding $\alpha$.

\subsection{Results}
We present our formal analysis results in Table~\ref{tab:formal}, with the detailed derivation and computing steps of these upper bounds deferred to~\ref{app:formal} for brevity. In the following, we will focus on discussing key observations from the formal analysis result.
\begin{table}[]
\caption{Formal comparison between DP-PLR and existing DP indexes}
\vspace{2mm}
\label{tab:formal}
\scalebox{0.72}{
\begin{tabular}{ccccc}
\hline
\multicolumn{1}{|l|}{}                               & \multicolumn{1}{c|}{\textbf{DP $B^{+}$ Tree}}                                         & \multicolumn{1}{c|}{\textbf{Crypt$\epsilon$}}                                & \multicolumn{1}{c|}{\textbf{SPECIAL}}                                             & \multicolumn{1}{c|}{\textcolor{blue}{\textbf{DP-PLR (ours)}}}                                                                                        \\ \hline
\multicolumn{1}{|c|}{\textbf{Query error}}           & \multicolumn{1}{c|}{$O\left(\frac{2\log{N}{\ln(\frac{2}{\beta})}}{\epsilon}\right)-B^{\ddagger}$} & \multicolumn{1}{c|}{$O\left(\frac{2N{\ln(\frac{2}{\beta})}}{\epsilon}\right)$} & \multicolumn{1}{c|}{0}                                                            & \multicolumn{1}{c|}{\textcolor{blue}{$O\left(\frac{(2\log{N})^{\frac{3}{2}}\sqrt{\ln(\frac{2}{\beta})}}{\epsilon}\right) - \alpha_{s}\frac{(\log{N})^{\frac{3}{2}}}{\epsilon}$}} \\
\multicolumn{1}{|c|}{\textbf{Query overhead}}        & \multicolumn{1}{c|}{$O\left(\frac{2\log{N}{\ln(\frac{2}{\beta})}}{\epsilon}\right)+B$} & \multicolumn{1}{c|}{$O\left(\frac{2N{\ln(\frac{2}{\beta})}}{\epsilon}\right)$} & \multicolumn{1}{c|}{$O\left(\frac{4N{\ln(\frac{1}{\beta})}}{\epsilon}\right)+\mu$} & \multicolumn{1}{c|}{\textcolor{blue}{$O\left(\frac{(2\log{N})^{\frac{3}{2}}\sqrt{\ln(\frac{2}{\beta})}}{\epsilon}\right) + 2\alpha_{s}\frac{(\log{N})^{\frac{3}{2}}}{\epsilon}+2\tau$}}                 \\
\multicolumn{1}{|c|}{\textbf{Index size (bits)}}     & \multicolumn{1}{c|}{$64(2N-1)$}                                                       & \multicolumn{1}{c|}{$64N$}                                                   & \multicolumn{1}{c|}{$128N$}                                                       & \multicolumn{1}{c|}{\textcolor{blue}{$128M$ or $O(1)$ when $M$ is small}}                                                                                                 \\
\multicolumn{1}{|c|}{\textbf{Data overhead}} & \multicolumn{1}{c|}{$O\left(\frac{2\log{N}\sqrt{N\ln(\frac{2}{\beta})}}{\epsilon}\right)+NB$} & \multicolumn{1}{c|}{0}                                                       & \multicolumn{1}{c|}{0}                                                            & \multicolumn{1}{c|}{\textcolor{blue}{0}}                                                                                                      \\
\multicolumn{1}{|c|}{\textbf{Privacy}}     & \multicolumn{1}{c|}{$(\epsilon)$-DP}                                              & \multicolumn{1}{c|}{$(\epsilon)$-DP}                                     & \multicolumn{1}{c|}{$(\epsilon, \delta)$-DP}                                          & \multicolumn{1}{c|}{\textcolor{blue}{$(\epsilon)$-DP}}                                                                                    \\ \hline
\multicolumn{5}{l}{$^{\ddagger}$. $B$ denotes the overflow array size, which denotes the fixed number of extra dummy injected to each leaf node data~\citep{sahin2018differentially}.}         
\end{tabular}}
\vspace{-1em}
\end{table}

{\bf Observation 1. Our DP-PLR offers significantly lower storage costs, potentially constant in size, for the index structure compared to existing DP indexes, which typically require storage linear to the key size \( N \). } The storage costs of existing DP indexes (all baselines) are linear in the key size, \( O(N) \), which can be substantial for attributes with large domain sizes (e.g., salary, mortgage). In contrast, the storage cost of DP-PLR is determined by the number of PLR models, \( M \). Notably, we can adjust the error bound \( \tau \) during PLR training or enforce an upper limit on \( M \) to control the maximum number of models. Thus, it is reasonable to assume that \( M \sim O(1) \). This compact storage size also results in faster lookup times. For example, given \( x_i \), the lookup time in DP-PLR is \( O(\log{M}) \sim O(1) \), assuming a binary search to locate the segment linear model for predicting \( x_i \)'s position. In contrast, the other three methods either require searching an array of size \( N \) (Crypt$\epsilon$ and SPECIAL) or traversing a tree with \( N \) leaf nodes to locate \( x_i \), making their lookup time no better than \( O(\log{N}) \).

{\bf Observation 2. DP-PLR is the only method that achieves (nearly) lossless indexing without requiring the insertion of dummy data, while still maintaining pure \( \epsilon \)-DP guarantees.} From the indexing error bound, we observe that when \( \alpha_s \) is properly set, for instance \( \alpha_{s} \geq O(\sqrt{\ln(\beta^{-1})}) \), the probability of achieving no indexing error is at least \( 1-\beta \) with DP-PLR. This near-lossless guarantee (with high probability) is possible because DP-PLR accurately estimates the error bounds for both the noisy CFC and the PLR prediction, smoothing these errors with pessimistic overestimation. While the DP \( B^+ \) tree has similar potential if the overflow array size \( B \) is set large enough, it requires a significant amount of dummy data, at least proportional to \( O(\sqrt{N}\log{N}) \). In contrast, DP-PLR avoids the need to inject any dummy data. Moreover, while SPECIAL achieves a deterministic lossless guarantee, it uses one-sided Laplace noise, which limits it to providing only \( (\epsilon, \delta) \)-DP.

{\bf Observation 3. DP-PLR shows significantly lower indexing overhead compared to Crypt$\epsilon$ and SPECIAL, and is only marginally larger than DP $B^{+}$ by an asymptotic factor of \( O(\sqrt{\log{N}}) \). However, DP-PLR eliminates the need to inject any dummy data into the underlying data.} When setting \( \alpha_s = O(\sqrt{\ln{(\beta^{-1})}}) \) and \( B \) proportional to \( O(2\epsilon^{-1}\log{N}{\ln(\beta^{-1})}) \), both DP $B^{+}$ and DP-PLR provide the same accuracy guarantee (i.e., no error with probability at least \( 1-\beta \)). When comparing their overhead, DP-PLR has a slightly larger asymptotic factor of \( O(\sqrt{\log{N}}) \) than DP $B^{+}$. However, in terms of data storage overhead, DP $B^{+}$ is \( O(\log{N}\sqrt{{N}}) \) larger than DP-PLR. When compared to other noisy CFC-based DP indexes, such as Crypt$\epsilon$ and SPECIAL, these methods have an overhead of at least \( O(N/\log^{3/2}{N}) \) relative to DP-PLR. This demonstrates that DP-PLR also minimizes overhead compared to state-of-the-art noisy CFC-based DP indexes.

\vspace{-1em}
\section{Conclusion}
\vspace{-1em}
In this work, we propose DP-PLR, the first differentially private (DP) learned index designed to efficiently process predicate queries on encrypted databases. Our formal analysis demonstrates that DP-PLR significantly reduces index storage costs, shifting from the traditional linear dependence on key sizes to potentially constant size. Furthermore, to the best of our knowledge, DP-PLR is the only approach that achieves (nearly) lossless indexing without requiring the insertion of dummy data, while maintaining strict $\epsilon$-DP guarantees. Its query overhead also exhibits asymptotically better performance compared to the current state-of-the-art lossless private index, SPECIAL. Additionally, we identify several exciting opportunities for future optimization, such as enabling private and efficient updates to DP-PLR, supporting concurrent yet private lookups, and leveraging available public data to fine-tune the index for improved accuracy and reduced query overhead.

\bibliography{iclr2024_conference}
\bibliographystyle{iclr2024_conference}

\appendix
\section{Appendix}
\subsection{Formal analysis}\label{app:formal}
{\bf DP $B^{+}$ tree indexes~\citep{sahin2018differentially}.} Given an index lookup with attribute value $k$, the DP $B^{+}$ method will return (pointers to) all tuples linked to the $k$ matching leaf nodes. Since Laplace noise $Z$ in the scale of $Z \sim \texttt{Lap}(\log{N}/\epsilon)$ is added to distort number of records to be returned. So that we know that $Z$ will be the primary factor that causes query error (missing tuples due to negative values of $Z-B < 0$) or overhead ( when $Z-B >0$). Hence, the following analysis will be primarily based on the upper bounds of $Z$. Since \( Z \sim \texttt{Lap}(\log{N}/\epsilon) \), and thus the probability density function (PDF) of \( Z \) is \( f_Z(z) = \frac{1}{2b} \exp\left(-\frac{|z|}{b}\right) \) for \( z \in \mathbb{R} \). The tail probability of the Laplace distribution is given by \( P(|Z| > \alpha) = \exp\left(-\frac{\alpha}{2b}\right) \). When setting \( \exp\left(-\frac{\alpha}{b}\right) \leq \beta \), and taking the natural logarithm of both sides, one can obtain \( \alpha \geq -b \ln(\beta) \). Substituting \( b = \frac{\log N}{\epsilon} \) gives:
\[
\alpha \geq \frac{\log N}{\epsilon} \ln\left(\frac{1}{\beta}\right).
\]
By symmetry of the Laplace noises, we know that with probability at least $\beta$ such that $Z > \frac{\log N}{\epsilon} \ln\left(\frac{2}{\beta}\right)$ or $Z < -\frac{\log N}{\epsilon} \ln\left(\frac{2}{\beta}\right)$. So that we can derive the query error bound as $\max(0, -(Z+B))$, which is $O\left(\frac{\log N \ln\left(\frac{2}{\beta}\right)}{\epsilon} \right)-B$, and similarly the overhead is $O\left(\frac{\log N \ln\left(\frac{2}{\beta}\right)}{\epsilon} \right) + B$.

Next, we analyze the total storage overhead of the DP $B^{+}$ tree indexes. So in general, the total overhead is determined by $Z_1 + ... + Z_N$, where each $Z_i$ denotes the noise added to the $i^{th}$ leaf node for populating the tree node storage. To derive the upper bounds, here we need the Lemma 10 in~\citet{wang2022incshrink} (or Lemma 12.2 in~\citet{dwork2014algorithmic}). So for completeness, we reproduce the lemmas as below:
\begin{lemma}\label{lemma:sum}
Given $M$ i.i.d. Laplace random variables, $Z_1,...,Z_N$, where each $Y_i$ is sampled from $\textup{Lap}(\frac{\Delta}{\epsilon})$. Let $0 < \alpha \leq N\frac{\Delta}{\epsilon}$, the following inequality holds
$$\textup{Pr}\left[~ \sum_{1}^{N}{Z_i} \geq \alpha \right] \leq e^{\left( \frac{-\alpha^2\epsilon^2}{4N\Delta^2} \right)}$$
\end{lemma}
\begin{proof}The complete proof of Lemma~\ref{lemma:sum} can be found in the Appendix C.1 of~\cite{wang2021dp} and in~\cite{dwork2014algorithmic}. 
\end{proof}
By setting $e^{\frac{-\alpha^2\epsilon^2}{4M\Delta^2}} = {\beta}$, and when $M > 4\ln{\frac{1}{\beta}}$ the following inequality holds
$$\textup{Pr}\left[~ \sum_{i=1}^{N} Z_i \geq 2\frac{\Delta}{\epsilon}\sqrt{N\ln{\frac{1}{\beta}}} ~\right] \leq \beta $$
Hence, we can derive the total storage overhead upper bound as $O\left(\frac{2\log N \sqrt{N\ln{\frac{2}{\beta}}}}{\epsilon} \right) + NB$

The DP $B^{+}$ tree has a total of $2N - 1$ nodes, where each node needs to store at least an integer\footnote{We assume 64 bits for all integers and address spaces}. Thus, the index storage requires $64(2N - 1)$ bits.

{\bf Crypt$\epsilon$ indexes~\citep{chowdhury2019crypt}.} The Crypt$\epsilon$ indexes use the noisy CFC model, so they directly release $\{x_i, \tilde{y}_i\}_{i=1}^{N}$. Errors occur when $\tilde{y}_i < y_i$, determined by $y_i - \tilde{y}_i = -\texttt{Lap}(N/\epsilon)$. Using the same upper bound technique as the DP $B^{+}$ tree, both the query error and overhead of Crypt$\epsilon$ indexes are within $O\left(\frac{2N \ln\left(\frac{2}{\beta}\right)}{\epsilon} \right)$. As Crypt$\epsilon$ store only the noisy CFC, the index storage is $64N$ bits, and there is no data storage overhead, as it does not use dummy data.

{\bf SPECIAL indexes~\citep{wang2024special}.} The SPECIAL indexes use the same noisy CFC model as Crypt$\epsilon$, but release two CFCs with one-sided Laplace noise—shifted to have mean $\mu > 0$ for positive (or $\mu < 0$ for negative) values and truncated at 0. This method achieves $(\epsilon, \delta)$-DP instead of $\epsilon$-DP. As SPECIAL consistently overestimate the indexing range, no query errors occur, but query overhead may arise. Since the Laplace noises are with shifted means, and both CFC can lead to overheads. Since SPECIAL introduces two noisy curves, so that the overhead is bounded by $O\left(\frac{4N \ln\left(\frac{2}{\beta}\right)}{\epsilon}\right) + \mu $. The index storage is doubled compared to the Crypt$\epsilon$ with $128N$ bits. Similarly, there is no data storage overhead.

{\bf DP-PLR Indexes (Our Method).} In our DP-PLR indexes, we will first look into the privacy induced errors in the binary mechanism for releasing noisy CFCs. Since each noisy CFC data point is released by using at most $\log N$ nodes, where each node is distorted by Laplace noise in the scale of $\texttt{Lap}(\log{N}/\epsilon)$, so that by Lemma~\ref{lemma:sum}, we can conclude that the privacy induced error bound of binary mechanism (for estimating each CFC data points) is within $O\left(\frac{2\log N \sqrt{\log N\ln{\frac{2}{\beta}}}}{\epsilon} \right)$, and further we can derive the query error bound as $O\left(\frac{2\log N \sqrt{\log N\ln{\frac{2}{\beta}}}}{\epsilon} \right) - \alpha_s \frac{\log^{3/2}N}{\epsilon}$. Note that the PLR process has bounded estimation error in $\tau$, and in the indexing, we use $\tau$ to carry out pessimistic overestimation, so that the query error bounds will not be affected but the overhead now becomes $O\left(\frac{2\log N \sqrt{\log N\ln{\frac{2}{\beta}}}}{\epsilon} \right) + 2\alpha_s \frac{\log^{3/2}N}{\epsilon} + 2\tau$, consider the overestimation is applied to both end points. The index storage for PLR is bounded by \( 2\times64\times M \) bits, where \( M \) is the number of linear models in the PLR. Similarly, DP-PLR does not use dummy data, resulting in no additional data storage overhead.

\end{document}